\documentclass[conference]{IEEEtran}
\IEEEoverridecommandlockouts
\usepackage{booktabs} 
\usepackage{graphicx}
\usepackage{amsthm}
\usepackage{amsmath}
\usepackage{tcolorbox}
\usepackage[ruled,vlined,linesnumbered]{algorithm2e}
\usepackage{algpseudocode}
\usepackage{amssymb}
\usepackage{array}
\usepackage{tabularx}
\usepackage{url}
\usepackage{lipsum}
\usepackage{multirow}
\usepackage[noadjust]{cite}

\newtheorem{theorem}{Theorem}[section]

\def\BibTeX{{\rm B\kern-.05em{\sc i\kern-.025em b}\kern-.08em
    T\kern-.1667em\lower.7ex\hbox{E}\kern-.125emX}}
\begin{document}

\title{LEP-CNN: A Lightweight Edge Device Assisted Privacy-preserving CNN Inference Solution for IoT}

\author{\IEEEauthorblockN{ Yifan Tian and Jiawei Yuan}
tiany1@my.erau.edu, yuanj@erau.edu
\IEEEauthorblockA{\textit{Department of ECSSE}} 
\textit{Embry-Riddle Aeronautical University, USA}
\and
\IEEEauthorblockN{Shucheng Yu}
Shucheng.Yu@stevens.edu
\IEEEauthorblockA{\textit{Department of ECE}} 
\textit{Stevens Institute of Technology, USA}
\and
\IEEEauthorblockN{Yantian Hou}
yantianhou@boisestate.edu
\IEEEauthorblockA{\textit{Department of ECSSE}} 
\textit{Boise State University, USA}
}

\maketitle

\begin{abstract}
Supporting convolutional neural network (CNN) inference on resource-constrained IoT (Internet of Things) devices in a timely manner has been an outstanding challenge for emerging smart systems. To mitigate the burden on IoT devices, one prevalent solution is to offload the CNN inference task, which is usually composed of billions of operations, to public cloud. However, the ``offloading-to-cloud'' solution may cause privacy breach while moving sensitive data to cloud. For privacy protection, the research community has resorted to advanced cryptographic primitives (e.g., homomorphic encryption) and approximation techniques to support CNN inference on encrypted data. Consequently, these attempts cause impractical computational overhead on IoT devices and degrade the performance of CNNs. Another concern of ``offloading-to-cloud'' besides the privacy issue is the integrity of data. In order to avoid the heavy computation resulted by the offloaded tasks, public cloud can dismiss the inference request by sending back random results to IoT devices. Moreover, relying on the remote cloud can cause additional network latency and even make the system dysfunction when network connection is off. 

To address the challenge, we proposes an extremely \underline{l}ightweight \underline{e}dge device assisted \underline{p}rivate CNN inference solution for IoT devices, namely LEP-CNN. The main design of LEP-CNN is based on a novel online/offline encryption scheme. The decryption of LEP-CNN is pre-computed offline via utilizing the linear property of the most time-consuming operations of CNNs. As a result, LEP-CNN allows IoT devices to securely offload over 99\% CNN operations, and edge devices to execute CNN inference on encrypted data as efficient as on plaintext. To prevent edge devices from returning an incorrect result, LEP-CNN provides an integrity check option to enable the IoT devices to detect incorrect results with a successful rate over 99\%. Experiments on AlexNet show that our scheme can speed up the CNN inference for more than $35\times$ for IoT devices with comparable computing power of Raspberry Pi. We also implemented a homomorphic encryption based privacy preserving AlexNet using the well-known CryptoNets scheme and compared it with LEP-CNN. We demonstrated that LEP-CNN has a better performance than homomorphic encryption based privacy preserving neural networks under time-sensitive scenarios.

\end{abstract}


\section{Introduction}
With the recent advances in artificial intelligence, the integration of deep neural networks and IoT is receiving increasing attention from both academia and industry \cite{dl-iot-1,dl-iot-2,dl-iot-5}. As the representative of deep neural networks, convolutional neural network (CNN) has been identified as an emerging technique to enable a spectrum of intelligent IoT applications \cite{dl-iot-survey}, including visual detection, smart security, audio analytics, health monitoring, infrastructure inspection, etc. However, due to the high computational cost introduced by CNNs, their deployment on resource-constrained IoT devices for time-sensitive services becomes very challenging. For example, popular CNN architectures (e.g., AlexNet \cite{AlexNet}, FaceNet \cite{FaceNet}, and ResNet \cite{ResNet}) for visual detection require billions of operations for the execution of a single inference task. Our evaluation results show that an inference task using AlexNet can cost more than two minutes on an IoT device with comparable computing capability as a Raspberry Pi (Model A). To mitigate such a burden for IoT devices, offloading CNN tasks to public cloud has become a popular choice in the literature. However, this type of ``cloud-backed" system may raise privacy concerns \cite{iot-cloud-privacy} by sending sensitive data to remote cloud. Also, the integrity of the returned results cannot be guaranteed when the cloud is ``curious-and-dishonest''. Moreover, connecting to cloud can cause additional latency to the system under network congestion and even make the system dysfunction when network is off \cite{edge-computing}.

In time-sensitive CNN-driven IoT applications, the inference stage of a trained CNN is executed for deep data analytics. To address the privacy concern with offloading CNNs, researchers have attempted to execute the inference stage over encrypted data \cite{CryptoNets,IACR-PPDL,CryptoDL}. These schemes usually first use approximation strategies to convert non-linear layers in a CNN to linear operations. Then, homomorphic encryption is utilized to enable privacy-preserving execution of these converted operations and other layers in the CNN using cloud computing. Nevertheless, the adoption of homomorphic encryption introduces extremely high encryption cost and/or communication load to the local IoT devices. For example, a quad-core Raspberry Pi, which outperforms most IoT devices in terms of computational capability, can perform only four Paillier homomorphic encryption per second \cite{paillier-benchmark}. Given a single input of AlexNet that has $227\times227\times 3$ elements, it requires more than 10 hours to complete the encryption, which is impractical for most applications in terms of both time delay and energy consumption. (As a comparison, to execute an inference task with the same AlexNet, a single-core Raspberry Pi with our solution can finish all local encryption and decryption within 0.3 second as shown in Section \ref{s:evaluation}.) Besides the high computational cost, the utilization of approximation in existing research also results in accuracy loss to some extent. Furthermore, these research adopt batch processing to improve their performance, which is more suitable for the ``Data Collection and Post-Processing'' routine. Differently, on-the-fly processing is desired for IoT devices to fulfill time-sensitive tasks. Another line of related research utilizes differential privacy to achieve privacy-preserving offloading of the training stage of CNNs \cite{CCS16,ICDM17-PHAN}. These research control the amount of information leaked from each individual record in the training dataset. However,  differential privacy becomes unsuitable for the inference stage, because only a single input is available at this stage.


The other limitation with offloading time-sensitive IoT tasks to cloud is the reliance on the availability of cloud and its network condition. To overcome this limitation, one prevalent solution is to utilize the computing resources at the edge of network. Compared with cloud computing, edge devices are geographically closer to IoT devices and usually within one-hop communication range. Such physical proximity can effectively ameliorate the network latency and availability issue. Most current research \cite{edge-iot-2,edge-iot-3} on edge computing mainly emphasizes on fundamental issues such as resource allocation but assumes the edge devices are fully trusted. Such an assumption, while is probably true for many applications, can be problematic for other applications where IoT devices are mobile or edge devices are shared by ad hoc multiple parties that do not share mutual trust. Examples of such applications include vehicular networks, drones, and mobile edge computing (MEC) \cite{5G-MEC} in general, to name a few. Some research \cite{iot-survey} does address data security and privacy issues in edge computing. However, efficient and private offloading of CNN inference to edge devices remains an open challenge.  

Another issue is that the IoT devices are not able to judge the correctness of the returned results from the cloud. Since the offloaded computation is resource consuming, public clouds may not be willing to allocate expensive computational resources and may tend to cheat the IoT devices by returning random data with the same size of the desired data. According to \cite{kandukuri2009cloud}, in some cases, dishonest cloud service provides may even discard the data to save resources. How to effectively detect such dishonest behaviors while maintaining the lightweight computation on IoT devices and overall performance in time-sensitive CNN inference tasks is another essential challenge to be solved. 

This paper addresses such challenges and proposes an extremely lightweight edge device assisted private CNN inference solution for IoT devices, namely LEP-CNN. LEP-CNN enables IoT devices to securely offload CNN inference tasks to local edge devices. To significantly speed up the offloading process, LEP-CNN adopts a novel online/offline encryption. Specifically, since CNN linear operation over input data and random noise are separable, encryption and decryption can be efficiently computed offline. In practical CNN architectures such as AlexNet and FaceNet, linear operations are dominant due to their vast number. Therefore, it is rewarding to trade offline computation and storage (of random noise) for online computation. As a result, our online/offline encryption allows IoT devices to securely offload over 99\% CNN operations to edge devices. And edge devices are able to execute CNN inference on encrypted data as efficiently as on unencrypted data. In addition, LEP-CNN does not introduce any accuracy loss as compared to CNN inference over unencrypted data. Compared with homomorphic encryption based privacy preserving neural networks \cite{CryptoNets,CryptoDL,IACR-PPDL}, LEP-CNN achieves a better performance under time-sensitive scenarios due to its lightweight property. Furthermore, to detect dishonest behaviors from ``curious-and-dishonest'' edge devices, our scheme provides an integrity check mechanism which helps the IoT devices detect incorrect returned results from edge devices with an over 99\% success rate. Minor computation overhead (1.1\% drop of offload percentage in worst case) is introduced when this integrity check is on. This integrity check can be optionally turned off in such scenario that the cloud is considered ``curious-and-honest'' to save local computational resources. LEP-CNN can be customized to support flexible CNN architectures that fulfill the requirements of different applications.


Extensive experimental evaluation shows the efficiency, scalability, accuracy and integrity of LEP-CNN. We implemented a prototype over well-known ImageNet \cite{imagenet_cvpr09} dataset using an uncompressed AlexNet architecture, which involves 2.27 billion operations for each inference result and has comparable complexity with these widely adopted architectures, e.g., FaceNet (1.6 billion operations) and Results (3.6 billion operations). The experimental results show that LEP-CNN can securely offload $99.95\%$ computation from the IoT devices for AlexNet. As a result, we are able to speed up $35.63\times$ for the execution of CNN-driven IoT AlexNet tasks using a single laptop as the edge device. Meanwhile, LEP-CNN saves over $95.56\%$ energy consumption compared with fully executing an AlexNet request on the IoT device. We also deployed LEP-CNN in a "curious-and-dishonest" scenario. With the integrity check feature being turned on, our scheme can still maintain a high computation offloading rate of $99.33\%$ from the IoT devices and a $30.00\times$ speedup. In addition, our scheme keeps the high speedup rate as the complexity of CNN layers increases. Therefore, it is promising to be scaled up for more complex CNN architectures. 

The rest parts of this paper are organized as follows: In Section \ref{s:cnn}, we introduce the background of CNN. Section \ref{s:construction} presents the detailed construction of our scheme. We state security analysis and numerical analysis in Section \ref{s:analysis}. We further evaluate the practical performance of our scheme with a prototype evaluation in Section \ref{s:evaluation}. We review and discuss related works in Section \ref{s:related-work} and conclude this paper in Section \ref{s:conclusion}. 


\section{Background of Convolutional Neural Network}\label{s:cnn}

A CNN contains a stack of layers that transform input data to outputs with label scores. There are four types of most common layers in CNN architectures, including: \emph{Convolutional Layers, Pooling Layers, Activation Layers}, and \emph{Fully-connected Layers}.


Convolutional layers extract features from input data. Fig.\ref{f:con-layer} depicts an example of convolutional layer that has an input data of size $n\times n\times D$ and $H$ kernels, each of size $k\times k \times D$. The input will be processed into all $H$ kernels independently to extract $H$ different features. Considering the input and each kernel as $D$ levels, where each level of the input and kernel are a $n\times n$ matrix and a $k\times k$ matrix respectively. Each level of a kernel starts scanning the corresponding input level from top-left corner, and then moves from left to right with $s$ elements, where $s$ is the stride of the convolutional layer. Once the top-right corner is reached, the kernel moves $s$ elements downward and scans from left to right again. This convolution process is repeated until the kernel reaches the bottom-right corner of this input level. For each scan, an output is computed using the dot product between the scanned window of input and the kernel as an example shown in Fig.\ref{f:con-layer}. For each kernel, the output for all $D$ levels will be summed together. 



\begin{figure}[ht]
\centering
\includegraphics[width=8cm]{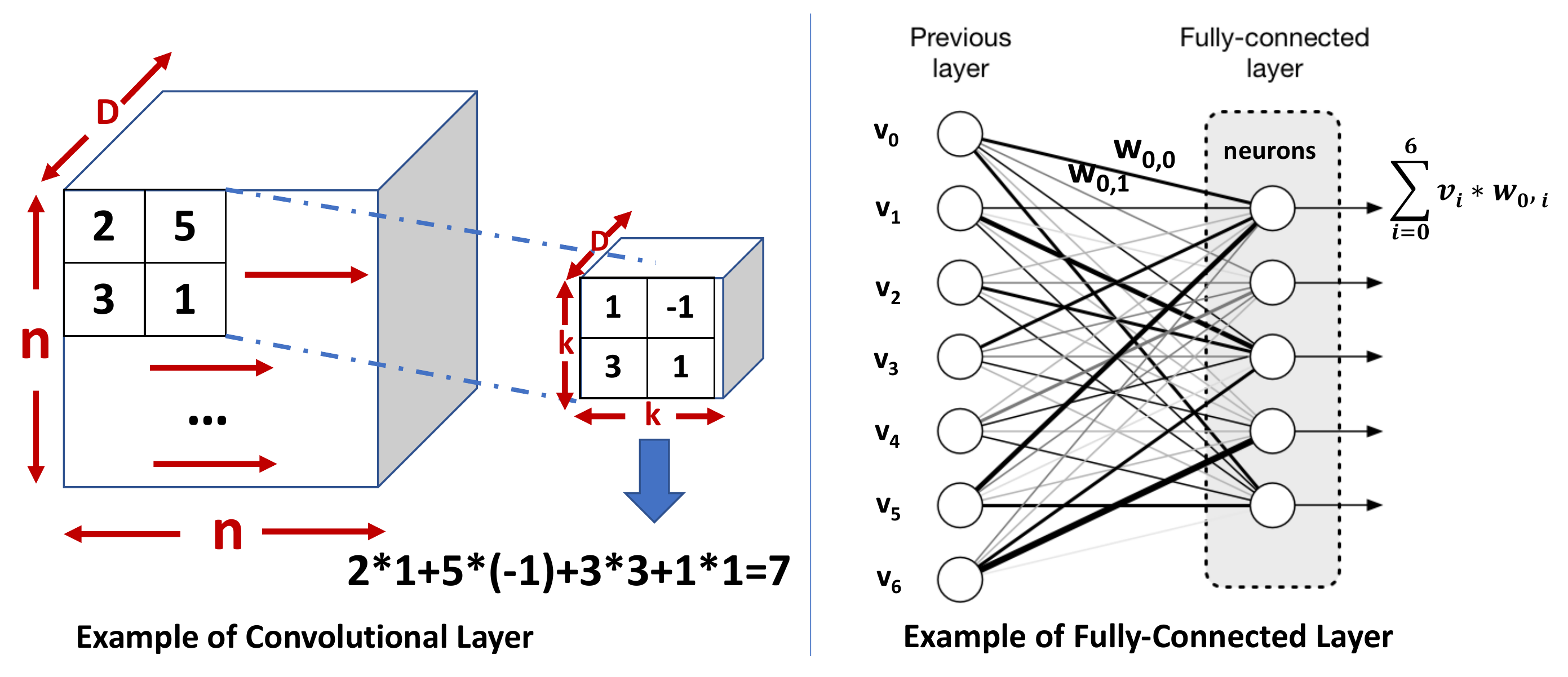}
\centering
\caption{Examples of a Convolutional Layer and a Fully-connected Layer} \label{f:con-layer}
\end{figure}

Pooling layers and activation layers are usually non-linear layers. A pooling layer is periodically inserted between convolutional layers. Pooling layers progressively reduce the spatial size of outputs from convolutional layers, and thus to make data robust against noise and control overfitting. An activation layer utilizes element-wise activation functions to signal distinct identification of their input data. There are a number of popular pooling strategies (e.g., max-pooling and average-pooling) and activation functions (e.g., rectified linear units (ReLUs) and continuous trigger functions), which are extremely computational efficient compared with convolutional layers and fully-connected layers. In our scheme, these two efficient layers will be directly handled on the IoT devices.


Fully-connected layers are usually the final layers of a CNN to output the final results of the network. In case of a fully-connected layer, all neurons in it have full connections to all outputs from the previous layer. As an example shown in Fig.\ref{f:con-layer}, the connection between each neuron and input element has a weight. To obtain the output of a neuron, elements connected to  it will be multiplied with their weights and then accumulated. 

More details about CNN can be found in ref \cite{cnn-wiki}.


\section{Detailed Construction of LEP-CNN}\label{s:construction}
\textbf{\textit{System Model}}: LEP-CNN consists of three major entities: \textit{IoT Devices}, the \textit{Owner of Devices}, and \textit{Edge Devices}. To deploy IoT devices for CNN-driven tasks, the owner first performs the offline phase to generate encryption and decryption keys and stores them into the IoT devices. As each set of encryption and decryption keys will only be used for one CNN request for security purpose, the owner needs to pre-load multiple sets of keys into the IoT devices. In Section \ref{ss:key-update}, we discuss the storage and remote update of keys to ensure the performance of different IoT application scenarios. For trained CNN architectures, the owner stores them using cloud storage or other accessible platforms, which can be retrieved by the edge devices when the IoT device enters the their coverage or pre-requested by the owner. In the online phase, once the IoT device has a piece of data needs to be processed by the CNN, it interacts with the edge device to perform privacy-preserving execution. During this process, the IoT device efficiently offloads expensive convolutional layers and fully-connected layers to the edge device, and only keeps the compute-efficient layers at local. This is motivated by the fact that convolutional and fully-connected layers occupy majority of computation and parameters storage in typical CNNs \cite{Cong2014}.

\textbf{\textit{Threat Model}}: In our system, we consider the edge devices to be ``curious-but-honest'', i.e., the edge devices follow our scheme to correctly conduct the storage, communication, and computational tasks, but try to learn sensitive information in IoT devices' input for the CNN. With our integrity check feature turned on, our system can also detect ``curious-and-dishonest'' edge devices at a high successful rate with minor efficiency trade-off. We assume the IoT devices are fully trusted and will not be compromised. The research on protecting IoT devices is orthogonal to this work. The edge device has access to the trained CNNs and all ciphertexts of inputs and outputs for the offloaded convolutional layers and fully-connected layers in the CNN. This assumption is consistent with majority of existing works that focus on privacy-preserving computation offloading \cite{CryptoNets,IACR-PPDL,CryptoDL}. LEP-CNN focuses on preventing the edge devices from learning the IoT devices' inputs and outputs of each layer in the offloaded CNN. 

There is also a line of research that proposes to infer the inputs of CNN layers from their outputs \cite{output-inference1}. This kind of inference is not applicable to LEP-CNN, since the edge device only has access to encrypted outputs from each layer.

We now present the detailed construction of LEP-CNN. Important notation is summarized in Table \ref{t:notations}.

\begin{table}[ht]
\caption{Summary of Notations}\label{t:notations}
\vspace{-0.2cm}
\begin{center}
\begin{tabular}{ |c |l |}
  \hline
   \multirow{2}{*}{$n\times n$}& the size of each level of a convolutional layer's\\   
   &input\\   \hline
   \multirow{2}{*}{$D$}&the depth of the input of a convolutional layer \\  
   &and kernel\\ \hline
   $k\times k$& the size of a convolutional layer's kernel matrix\\   \hline
   $H$& the number of kernels of a convolutional layer\\   \hline
   $s$& the size of stride used for a convolutional layer\\   \hline
   $p$& the size of padding used for a convolutional layer\\   \hline
   $\mathcal{R}_{c,d}$&$n\times n$ random matrices to encrypt the input of a \\
   $1\leq d\leq D$& convolutional layer  \\ \hline
  $\alpha_i$&$(\frac{n-k+2p}{s}+1)\times (\frac{n-k+2p}{s}+1)$ matrices to decrypt\\ 
   $1\leq i\leq H$& a convolutional layer's output from $H$ kernels\\\hline
   $m$& the size of a fully-connected layer input vector\\   \hline
   $T$& the number of neurons of a fully-connected layer\\   \hline
   \multirow{2}{*}{$\mathcal{R}_f$}&a $m$-dimensional random vector to encrypt the  \\ 
   & input of a fully-connected layer\\ \hline
   \multirow{2}{*}{$\beta$}&a $T$-dimensional decryption vector to decrypt  \\
   &fully-connected layer outputs \\ \hline
   $r$& the sample rate of returned data of a convolutional layer \\ \hline
   $\theta$& the error rate of returned data of a convolutional layer \\ \hline
\end{tabular}
\end{center}
\vspace{-0.7cm}
\end{table}

\subsection{Offline Phase}
In the offline phase, the owner generates encryption and decryption keys for all convolutional layers and fully-connected layers in a trained CNN. In LEP-CNN, we consider each element in the input data of convolutional layers and fully-connected layers is $\gamma$-bit long, and $\lambda$ is the security parameter. To ensure the security $\frac{1}{2^{\lambda-\gamma-1}}$ shall be a negligible value in terms of computational secrecy \cite{Crpto-book-cp3.3}, e.g., $<\frac{1}{2^{128}}$. Detailed selection of security parameter is discussed in Section \ref{ss:security-analysis}.



As described in Algorithm.\ref{a:off-covl}, given a convolutional layer with a $n\times n\times D$ input, stride as $s$, padding as $p$, and $H$ kernels ($k\times k$ matrices), the owner generates $\{\mathcal{R}_{c,d},1\leq d\leq D\}$ as the encryption keys and $\{\alpha_{i},1\leq i\leq H\}$ as its decryption keys, where $\mathcal{R}_{c,d}$ is a $n\times n$ random matrix and $\alpha_{i}$ is a $(\frac{n-k+2p}{s}+1)\times (\frac{n-k+2p}{s}+1)$ matrix. For expression simplicity, we use $\textbf{Conv}(\mathcal{R}_{c,d},i_{th})$ to denote the convolution operation for the $i_{th}$ kernel with $\mathcal{R}_{c,d}$ as input. 

Given a fully-connected layer with a $m$-dimensional vector as input and $T$ neurons, the owner first generates a $m$-dimensional random vector $\mathcal{R}_{f}$. Then, the owner takes $\mathcal{R}_{f}$ as the input of the fully-connected layer to output a $T$-dimensional vector $\beta$. $\mathcal{R}_f$ and $ \beta$ are set as the encryption key and decryption key respectively for this layer.

For a CNN with $x$ convolutional layers and $y$ fully-connected layers, $x$ sets of $\{\mathcal{R}_{c,d},\alpha_{i}\}_{1\leq d\leq D_x, 1\leq i\leq H_x}$ and $y$ sets of $\{\mathcal{R}_f, \beta\}$ are generated by the owner as a final set of keys $\{Enc_{key},Dec_{key}\}$. \textbf{\textit{Note that}}, \textit{each set of keys is only valid for one CNN request in the later online phase}. Thus, the owner will generate multiple sets of keys according to the necessity of application scenarios as discussed in Section \ref{ss:key-update}.

\begin{algorithm}[th]
\small
  \SetKwInOut{Input}{Input}
  \SetKwInOut{Output}{Output}

    \Input{Input size $n\times n \times D$, stride $s$, padding $p$, $H$ kernels}
    \Output{Encryption keys $\mathcal{R}_{c,d},1\leq d\leq D$, Decryption keys $\alpha_{i},1\leq i\leq H$}
     
    Generate random  $n\times n$ matrices $\mathcal{R}_{c,d},1\leq d\leq D$ \;
    \For{$1\leq i\leq H$}{
    	\For{$1\leq d\leq D$}{
      Take $\mathcal{R}_{c,d}$ as input for the $i_{th}$ kernel for convolution and output $\textbf{Conv}(\mathcal{R}_{c,d},i_{th})$\;
      d++;
      }
      Set $\alpha_{i}=\sum_{d=1}^D \textbf{Conv}(\mathcal{R}_{c,d},i_{th})$\;
      $i++$\;
  }
 \caption{Offline Preparation of Convolutional Layer}\label{a:off-covl}
\end{algorithm}

\begin{algorithm}[ht]
\small
  \SetKwInOut{Input}{Input}
  \SetKwInOut{Output}{Output}

    \Input{Input Data \& Trained CNN}
    \Output{CNN Execution Result}

    \textbf{Set} the Layer Input $\mathcal{M}$ = Input Data\;
    \textbf{Set} Layers = the collection of all Convolutional Layers and Fully-connected Layers in CNN\;
    \textbf{Set} Layer = the first Layer from Layers\;
    \While{Layer is not null}{
      \If{Layer $=$ Convolutional Layer} {Execute the \textbf{PPCL} with $\mathcal{M}$ as input.\;
      Set $\mathcal{M}$ = output from \textbf{PPCL}\;
      }

      \If{Layer $=$ Fully-connected Layer}
      {Execute the \textbf{PPFL} with $\mathcal{M}$ as input.

      Set $\mathcal{M}$ = output from \textbf{PPFL}\;
      }
      \textbf{Set} Layer = Layers.next()\;
  }
  \Return $\mathcal{M}$ as result\;
 \caption{Online Privacy-preserving CNN}\label{a:online-cnn}
\end{algorithm}

\subsection{Online Phase}
During the online phase, the IoT device can efficiently interact with the edge device to process data using CNN in a privacy-preserving manner. The overall process of our online phase is depicted in Algorithm.\ref{a:online-cnn}. Specifically, the IoT device offloads encrypted data to the edge device for performing compute-intensive convolutional layers and fully-connected layers. Intermediated results are returned back to the IoT device for decryption. Then, these decrypted results are processed with the follow up activation layer and pooling layer (if exist). Outputs are encrypted and offloaded again if the next layer is a convolutional layer or a fully-connected layer. This procedure is conducted iteratively until all CNN layers are executed. 

To fulfill these tasks, we designed two privacy-preserving schemes \textbf{PPCL} and \textbf{PPFL} for convolutional layers and fully-connected layers respectively. 

\subsubsection{PPCL: Privacy-preserving Convolutional Layer}\label{sss:ppcl}
In \emph{PPCL}, we consider a general convolutional layer with a $n\times n\times D$ input, stride as $s$, padding as $p$, and $H$ kernels with each size of $k\times k\times D$. The $d_{th}$ level of the input is denoted as a $m\times m$ matrix $\mathcal{I}_d$.

\textbf{Input Encryption}: The IoT device encrypts the input using the pre-stored keys $\{\mathcal{R}_{c,d}\}$ for this convolutional layer as 
\begin{eqnarray}\label{e:ppcl-enc}
Enc(\mathcal{I}_d)=\mathcal{I}_d+\mathcal{R}_{c,d}
\end{eqnarray} 
where $\{Enc(\mathcal{I}_d)\},1\leq d\leq D$ are sent to the edge device.  

\textbf{Privacy-preserving Execution}: The edge device takes each $Enc(\mathcal{I}_d),1\leq d\leq D$ as the input of kernels to perform the convolution process. For the $i_{th}$ kernel, the edge device outputs  
\begin{eqnarray}
&&\sum_{d=1}^D \textbf{Conv}(Enc(\mathcal{I}_d),i_{th})\\
&&=\sum_{d=1}^D \textbf{Conv}(\mathcal{I}_d,i_{th})+\sum_{d=1}^D \textbf{Conv}(\mathcal{R}_{c,d},i_{th})\nonumber
\end{eqnarray}
$\sum_{d=1}^D \textbf{Conv}(Enc(\mathcal{I}_d),i_{th}),1\leq i\leq H$ are returned back to the IoT device as intermediate results. 




\textbf{Decryption and Preparation for the Next Layer}: Given the returned $\sum_{d=1}^D \textbf{Conv}(Enc(\mathcal{I}_d),i_{th}),1\leq i\leq H$, the IoT device quickly decrypts them as 
\begin{eqnarray}
\sum_{d=1}^D \textbf{Conv}(Enc(\mathcal{I}_d),i_{th})-\alpha_i=\sum_{d=1}^D \textbf{Conv}(\mathcal{I}_d,i_{th})
\end{eqnarray}
where $\{\alpha_i=\sum_{d=1}^D \textbf{Conv}(\mathcal{R}_{c,d},i_{th})\}, 1\leq i\leq H$ are the pre-stored decryption keys for this layer. Afterwards, the IoT device performs the activation layer and pooling layer directly over convolutional output, which are extremely compute-efficient. For example, one of the most popular activation layer ReLU only requires translating negative values in the input to 0. The popular max-pooling (or average-pooling) layer simply shrinks the data by taking the max value (or average value respectively) every few values. The output will be encrypted and sent to the edge device using \textit{PPCL} for the next convolutional layer (or \textit{PPFL} respectively for a fully-connected layer). 


\subsubsection{PPFL: Privacy-preserving Fully-connected Layer}

In \emph{PPFL}, we consider a general fully-connected layer with $T$ neurons and takes a $m$-dimensional vector $\mathcal{V}$ as input.

\textbf{Input Encryption}: Given the input, the IoT device encrypts it using the pre-stored encryption key $\mathcal{R}_f$ for this layer as 
\begin{eqnarray}\label{e:ppfl-enc}
Enc(\mathcal{V})=\mathcal{V}+\mathcal{R}_f
\end{eqnarray} 
$Enc(\mathcal{V})$ is then sent to the edge device. 

\textbf{Privacy-preserving Execution}: On receiving $Enc(\mathcal{V})$, the edge device takes $Enc(\mathcal{V})$ as the input of the fully-connected layer. Specifically, the encrypted outcome $Enc(\mathcal{O}[j]),1\leq j\leq T$ of each neuron is computed as 
\begin{eqnarray}
Enc(\mathcal{O}[j])&=&\sum_{i=1}^m Enc(\mathcal{V})[i]\times w_{i,j}=\mathcal{O}[j]+\beta[j]
\end{eqnarray}
where $w_{i,j}$ is the weight between the $i_{th}$ element of input vector and the $j_{th}$ neuron. $Enc(\mathcal{O})=\{Enc(\mathcal{O}[1]),Enc(\mathcal{O}[2]),\cdots, Enc(\mathcal{O}[T])\}$ is sent back to the IoT device as intermediate results.

\textbf{Decryption and Preparation for the Next Layer}: Given the returned $Enc(\mathcal{O})$, the IoT device decrypts each $Enc(\mathcal{O})$  with the pre-stored decryption key $\beta$ of this layer as
\begin{eqnarray}\label{e:ppfl-dec}
\mathcal{O} = Enc(\mathcal{O}) - \beta
\end{eqnarray}
Then, the IoT device executes the activation layer with $\mathcal{O}$ as input. The output from the activation layer will be encrypted and sent to edge device using \textit{PPFL} if there are any additional fully-connected layers in the CNN.

To this end, the IoT device is able to efficiently handle each layer in a CNN. Compute-intensive convolutional and fully-connected layers are securely offloaded to the edge using \textit{PPCL} and \textit{PPFL}. These compute-efficient layers are directly handled by the IoT device. Since we develop \textit{PPCL} and \textit{PPFL} as independent modules, they can be customized and recursively plugged into any CNN no matter how many different convolutional layers and fully-connected layers it contains.

\subsection{Discussion - Storage and Update of Pre-computed Keys}\label{ss:key-update}
LEP-CNN considers two major types of resource-constrained IoT devices that run CNN-driven applications. 
\begin{itemize}
	\item Type-1: Mobile IoT devices with limited battery life and computational capability, such as drones. 
	\item Type-2: Static devices with power supply but has limited computational capability, such as security cameras.
\end{itemize}
The type-1 devices are usually deployed to perform tasks for a period time. Therefore, before each deployment, the device owner can pre-load enough keys to support its CNN tasks. With regards to the type-2 devices, the owner can perform an initial key pre-loading and then use remote update to securely add new offline keys as described in Fig.\ref{f:key-update}.

\begin{figure}[ht]
\centering
\includegraphics[width=8cm]{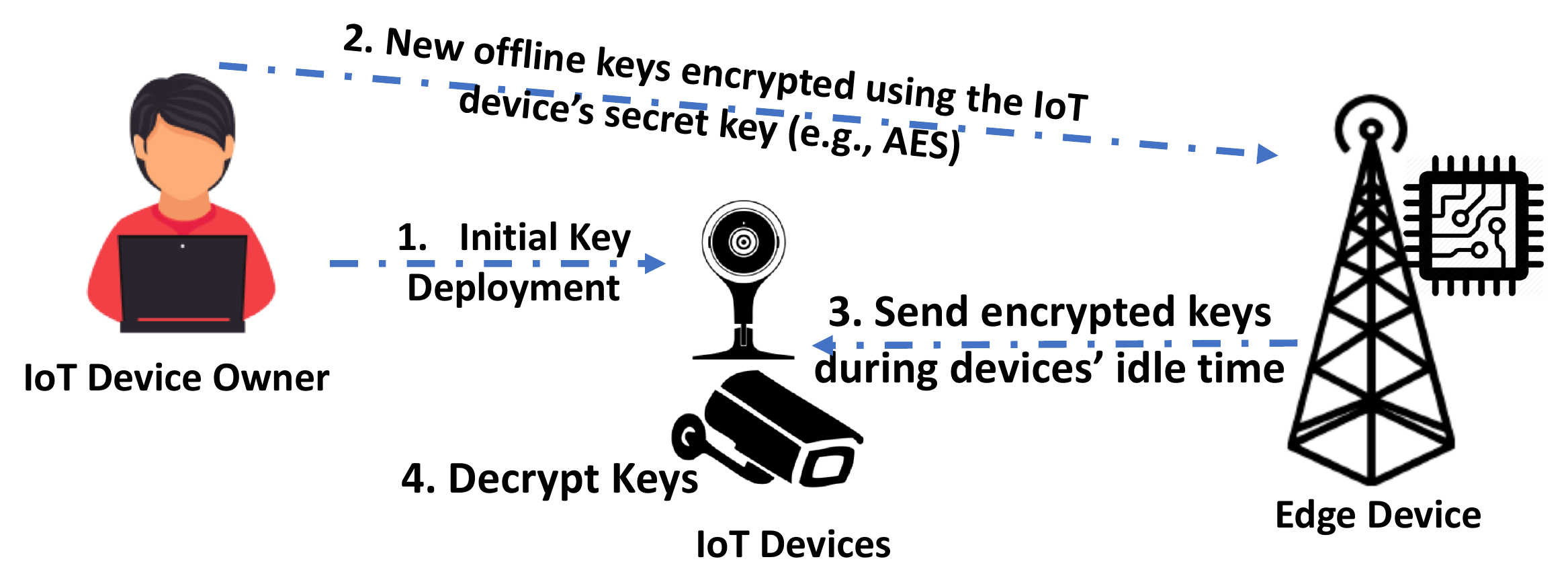}
\centering
\caption{Key Update for Power Connected Devices} \label{f:key-update}
\end{figure}

LEP-CNN proposes to ensure the timely processing of CNN requests when they are needed on IoT devices. Instead of performing real-time CNN requests on every piece of data collected, resource-constrained IoT devices usually require in-depth analytics using CNN when specific signals are detected. Taking real-time search and monitoring using drones as an example application for type-1 devices, fast local processing will be first performed for data collected to get estimated results \cite{uav-cloud}. Once suspicious signs are detected in estimated results, CNN based analytics are further conducted for a small set of data (e.g., video frames with the detected suspicious object). Given the high efficiency of LEP-CNN, the performance of such CNN requests will be timely supported when enough pre-computed keys are available. For example, when the average frequency of CNN requests is every one per ten seconds for a drone, only 360 sets of pre-computed keys are needed for one-hour deployment, which is longer than most current drones' battery life \cite{uav-battery}. Security camera is an example of type-2 devices, which requires CNN-based analytics to extract detailed information only when alarm is triggered by motion or audio sensors of the camera. Similar to the drone case, LEP-CNN can timely support the peak CNN requests when suspicious signs are detected. 

Assuming the average frequency of CNN-required alarm in a security camera is one per 10 minutes, and each alarm requires 5 CNN requests, 720 sets of pre-computed keys are needed for one-day usage. As evaluated in Section \ref{ss:storage}, an IoT device with a 32GB SD card is able to store keys to support 1600 requests for AlexNet. Such a result indicates 4.4 deployments and a 2.22-day support for type-1 and type-2 devices respectively when using AlexNet. 



Note that, LEP-CNN is designed for low-cost resource-constrained devices that require timely processing of moderate (or low) frequent CNN requests. For application scenarios that involve a large number of constant CNN requests, e.g., security critical surveillance systems, computational powerful devices are suggested to handle CNN requests directly at local.

\subsection{Discussion - Integrity Check on Returned Data}\label{ss:integrity-check}
In a scenario where the edge devices are ``curious-and-dishonest'', those edge devices may perform dishonest behaviors so that they can save their resource utilities. After receiving inference requests from the IoT devices, the edge devices may cheat the IoT devices by skipping the heavy convolutional operations and sending back random, apparently not correct, results to the IoT devices. These incorrect results can badly effect or even completely ruin the final result of the entire CNN inference. In order to ensure the integrity of returned data from edge devices, LEP-CNN also provides an optional integrity check functionality with only a minor efficiency cost. By enabling the integrity check, the IoT device can achieve an error detection rate of $99\%$ while only losing $1.1\%$ in offload percentage in the worst case compared with LEP-CNN disabling integrity check. The users can decide whether to turn on this functionality based on the actual deployment scenario and the trustworthiness of the edge devices.

The basic strategy of integrity check is to first sample a small portion of elements from returned data in each layer and then check whether there is an incorrect result occurring in the selected elements. To validate the correctness of a single element, the IoT device needs to go through the corresponded convolution operations locally. Although resource consuming, with high probability, this validation process can block IoT devices from taking incorrect results into next layer. 

The error detection rate $Pr(ED)$ can be calculated as below:
\begin{eqnarray}
Pr(ED) = 1 - \frac{{(1-\theta)N \choose rN}}{{N \choose rN}}
\end{eqnarray}

where $N$, $r$ and $\theta$ is the size, sample rate and error rate of returned data in a convolutional layer. ${(1-\theta)N \choose rN}$ is the combination operation for selecting $r \times N$ elements from $(1-\theta) \times N$ elements. In order to increase the error detection rate while lowering the additional validation computation on IoT device, we provide detailed evaluation by performing numerical analysis and practical experiments in Section \ref{s:analysis} and Section \ref{s:evaluation} respectively. Based on our experiment results on AlexNet, LEP-CNN with integrity check turned on could achieve over 99\% error detection rate while maintaining a similar computation offload rate compared with LEP-CNN with integrity check turned off.

\section{Analysis of LEP-CNN}\label{s:analysis}

\subsection{Security Analysis}\label{ss:security-analysis}

\begin{theorem}\label{th-1}
Given the ciphertext $\mathcal{C}$ of a $\gamma$-bit random message $\mathcal{M}$ generated using $\textit{PPCL}$ or $\textit{PPFL}$ in LEP-CNN, the probability for a probabilistic polynomial time (PPT) adversary $\mathcal{A}$ to output a correct guess for $\mathcal{M}$ shall have
\begin{eqnarray}
Pr[(\mathcal{M}^*=\mathcal{M})|\mathcal{C}] - Pr[\mathcal{M}^*=\mathcal{M}] \leq \epsilon
\end{eqnarray}
where $\epsilon$ is a negligible value in terms of computational secrecy \cite{Crpto-book-cp3.3}, $\mathcal{M}^*$ is $\mathcal{A}$'s guess for $\mathcal{M}$, and $Pr[\mathcal{M}^*=\mathcal{M}]$ is the probability $\mathcal{A}$ makes a correct without ciphertext. Specifically, the corresponding ciphertext generated using $\textit{PPCL}$ or $\textit{PPFL}$ only introduces negligible additional advantages to $\mathcal{A}$ for making a correct guess of $\mathcal{M}$.
\end{theorem}

\begin{proof}
As shown in Eq.\ref{e:ppcl-enc}, each level of the input data $\mathcal{I}_d$ in \textit{PPCL} is encrypted by adding an \textit{one-time} random matrix $\mathcal{R}_{c,d}$. With regards to each $\gamma$-bit element in $\mathcal{I}_d$, it is encrypted by adding a $\lambda$-bit random number $\mathcal{R}_{c,d}[e]$, i.e., $Enc(\mathcal{I}_d[e])=\mathcal{I}_d[e]+\mathcal{R}_{c,d}[e]$. Similarly, each element $\mathcal{V}[e]$ in the input of \textit{PPFL} is encrypted with a random number $\mathcal{R}_{f}[e]$ as $Enc(\mathcal{V}[e])=\mathcal{V}[e]+\mathcal{R}_{f}[e]$. For expression simplicity, we use $\mathcal{M}$ to denote a $\gamma$-bit input element for \textit{PPCL} or \textit{PPFL}, $R$ is the random number to encrypt $\mathcal{M}$ as $\mathcal{C}=\mathcal{M}+R$. Note that, $R$ will be re-generated for each encryption using \textit{PPCL} or \textit{PPFL}. 

To make a correct guess of $\mathcal{M}$ without the ciphertext, the adversary $\mathcal{A}$ has $Pr[\mathcal{M}^*=\mathcal{M}]=\frac{1}{2^\gamma}$, where $\mathcal{M}^*$ is $\mathcal{A}$'s guess for $\mathcal{M}$. 

By given a ciphertext $\mathcal{C}$, there are $2^\gamma$ possible values for its plaintexts $\mathcal{M}$ if $2^\gamma\leq \mathcal{C} \leq 2^\lambda-2^\gamma$, because $\mathcal{C}$ has the same distribution as the random number $R$ \cite{Crpto-book-cp11}. Now, if the random number $R$ used for encryption is the range of $[2^\gamma,2^\lambda-2^\gamma]$, we have $Pr[(\mathcal{M}^*=\mathcal{M})|\mathcal{C}] = Pr[\mathcal{M}^*=\mathcal{M}] = \frac{1}{2^\gamma}$. When $\mathcal{C}<2^\gamma$ or $\mathcal{C}>2^\lambda$, we have $Pr[(\mathcal{M}^*=\mathcal{M})|\mathcal{C}]>1/2^\gamma$. This is because the total possible inputs are reduced to $\mathcal{C}$ or $\mathcal{C}-2^\lambda$ respectively. Fortunately, the probability $Pr[\mathcal{C}<2^\gamma]$ or $Pr[\mathcal{C}>2^\lambda]$ in LEP-CNN is negligible when appropriate security parameter is selected. To be specific, $\mathcal{C}<2^\gamma$ or $\mathcal{C}>2^\lambda$ can appear when $R<2^\gamma$ or $R>2^\lambda-2^{\gamma}$. As $Pr[R<2^\gamma]=Pr[R>2^\lambda-2^{\gamma}]=\frac{2^\gamma}{2^\lambda}$, we have 
\begin{eqnarray}
&&Pr[R<2^\gamma~or~R>2^\lambda-2^{\gamma}] \nonumber\\
&&=Pr[R<2^\gamma]+Pr[R>2^\lambda-2^{\gamma}]=\frac{1}{2^{\lambda-\gamma-1}} \nonumber
\end{eqnarray}
Thus, to guarantee $\frac{1}{2^{\lambda-\gamma-1}}$ is a negligible probability, such as $\frac{1}{2^{128}}$, LEP-CNN can set the security parameter $\lambda$ according to size of input message, i.e., $\lambda-\gamma-1>128$.

We now use $\epsilon=\frac{1}{2^{\lambda-\gamma-1}}$ to denote the negligible probability, and get the probability $Pr[(\mathcal{M}^*=\mathcal{M})|\mathcal{C}]$ as
\begin{eqnarray}
Pr[(\mathcal{M}^*=\mathcal{M})|\mathcal{C}]\leq \frac{1}{2^\gamma}*(1-\epsilon)+1*\epsilon = \frac{1}{2^\gamma}+(1-\frac{1}{2^\gamma})\epsilon \nonumber
\end{eqnarray}
where $\frac{1}{2^\gamma}*(1-\epsilon)$ is the probability for a correct guess for $2^\gamma\leq R\leq 2^\lambda-2^\gamma$, and the ``1'' in $1*\epsilon$ is best probability for a correct guess $\mathcal{A}$ can have when $[R<2^\gamma~or~R>2^\lambda-2^{\gamma}]$. As a result, we get 
\begin{eqnarray}
Pr[(\mathcal{M}^*=\mathcal{M})|\mathcal{C}] - Pr[\mathcal{M}^*=\mathcal{M}]\leq(1-\frac{1}{2^\gamma})\epsilon<\epsilon \nonumber
\end{eqnarray}
Since $\epsilon$ is negligible value, Theorem \ref{th-1} is proved.
\end{proof}

\begin{table*}
\caption{Numerical Analysis Summary}\label{t:na-summ}
\vspace{-0.6cm}
\begin{center}
\begin{tabular}{ | c | c | c | c | c | c | c |}
  \hline
  \multicolumn{7}{|c|}{\textbf{LEP-CNN}} \\ \hline
   & \multirow{2}{*}{\textbf{Input Size}} & \multicolumn{2}{c|}{\textbf{Computation of the IoT (FLOPs)}} & \textbf{Offloaded Cost}& \textbf{Communication}&\multirow{2}{*}{\textbf{Storage Overhead}}\\ \cline{3-4}
   &  & \textbf{Input} & \textbf{Results} &\textbf{to the Edge}&\textbf{Cost}&\\ 
    &  & \textbf{Encryption} & \textbf{Decryption} &\textbf{(FLOPs)}&(\textbf{Elements})&\textbf{(Elements)}\\ \hline
  \textbf{Convolutional}  & $n\times n \times D$ &$Dn^2$   &$H(\frac{n-k+2p}{s}+1)^2$&$2DHk^2(\frac{n-k+2p}{s}+1)^2$ &$Dn^2+H(\frac{n-k+2p}{s}+1)^2$&$Dn^2+$\\ 
   &  & &&&&$H(\frac{n-k+2p}{s}+1)^2$\\ \hline
  \textbf{Fully-connected}& $m$ & $m$  & $T$  & $2mT$ &$m+T$  &$m+T$ \\
   \hline
  \multicolumn{7}{|c|}{\textbf{Offloading using Plaintext without Privacy Protection}} \\ \hline \textbf{Convolutional}  & $n\times n \times D$ &N/A &N/A&$2DHk^2(\frac{n-k+2p}{s}+1)^2$ &$Dn^2+H(\frac{n-k+2p}{s}+1)^2$&N/A\\ 
   \hline
  \textbf{Fully-connected}& $m$ & N/A&N/A & $2mT$ &$m+T$  &N/A \\
   \hline
\end{tabular}
\end{center}
In this table: $s$ is the stride, $p$ is the size of padding,  $H$ is the number of kernels, $k\times k$ is the size of kernels of a convolutional layer; $T$ is the number of neurons of a fully-connected layer. Each element is 20 Bytes.
\vspace{-6mm}
\end{table*}

\begin{table*}
\caption{Numerical Analysis of Integrity Check}\label{t:Numerical-Analysis-Integrity-Check}
\vspace{-0.3cm}
\begin{center}
\begin{tabular}{ | c | c | c | c | c | c | c |}
  \hline
  \multicolumn{7}{|c|}{\textbf{LEP-CNN}} \\ \hline
   & \multicolumn{3}{c|}{\textbf{Computation of the IoT (FLOPs)}} & \textbf{Offloaded Cost}& \textbf{Communication}&\multirow{2}{*}{\textbf{Storage Overhead}}\\ \cline{2-4}
   & \textbf{Input} & \textbf{Results} & \textbf{Results} &\textbf{to the Edge}&\textbf{Cost}&\\ 
    & \textbf{Encryption} & \textbf{Decryption} & \textbf{Validation} &\textbf{(FLOPs)}&(\textbf{Elements})&\textbf{(Elements)}\\ \hline
  
  \textbf{Integrity} & \multirow{2}{*}{$Dn^2$} &$H(\frac{n-k+2p}{s}+1)^2$   &$2Dk^2\times$&$2DHk^2\times$ &$Dn^2+$&$Dn^2+Hk^2+$\\ 
   \textbf{Check}&  & &$\lceil rH(\frac{n-k+2p}{s}+1)^2 \rceil$&$(\frac{n-k+2p}{s}+1)^2$&$H(\frac{n-k+2p}{s}+1)^2$&$H(\frac{n-k+2p}{s}+1)^2$\\ \hline
  \textbf{No Integrity}& \multirow{2}{*}{$Dn^2$} & $H(\frac{n-k+2p}{s}+1)^2$   & $0$  & $2DHk^2\times$ &$Dn^2+$  &$Dn^2+$ \\
   \textbf{Check}&&&&$(\frac{n-k+2p}{s}+1)^2$&$H(\frac{n-k+2p}{s}+1)^2$&$H(\frac{n-k+2p}{s}+1)^2$ \\ 
   \hline
\end{tabular}
\end{center}
In this table: $s$ is the stride, $p$ is the size of padding,  $H$ is the number of kernels, $k\times k$ is the size of kernels of a convolutional layer; $\theta$ is the error rate of the returned data; $r$ is the sample rate of the returned data. Each element is 20 Bytes.
\vspace{-6mm}
\end{table*}

\begin{table*}
\caption{Example Numerical Analysis on AlexNet}\label{t:AlexNet}
\vspace{-0.2cm}
\begin{center}
\begin{tabular}{ |c |c |c | c |c | c |c |c |}
  \hline
  &\textbf{Parameters}&\textbf{Input Size}&\textbf{Computation}&\textbf{Offloaded Cost}&\textbf{Offloaded}&\textbf{Communication}&\textbf{Storage}\\
  &&&\textbf{of the IoT}&&\textbf{Percentage}&\textbf{Cost}&\textbf{Overhead}\\ \hline
  \multirow{2}{*}{\textbf{Conv-1}}&n=227, H=96   &                &                 &                     &         &            &            \\ 
                                  &k=11, s=4     & $227\times227\times 3$     & 444,987 FLOPs   & 210,830,400 FLOPs   & 99.79\% & 8691.15 KB & 8691.15 KB \\ \hline
  \multirow{2}{*}{\textbf{Conv-2}}&n=27, H=256   &             &                 &                     &         &            &            \\ 
                                  &k=5, s=1      & $27\times27\times 96$       & 256,608 FLOPs   & 895,795,200 FLOPs   & 99.97\% & 5011.88 KB & 5011.88 KB \\ \hline
  \multirow{2}{*}{\textbf{Conv-3}}&n=13, H=384   &             &                 &                     &         &            &            \\ 
                                  &k=3, s=1      & $13\times13\times 256$       & 108,160 FLOPs   & 299,040,768 FLOPs   & 99.96\% & 2112.50 KB & 2112.50 KB \\ \hline
  \multirow{2}{*}{\textbf{Conv-4}}&n=13, H=384   &              &                 &                     &         &            &            \\ 
                                  &k=3, s=1      & $13\times13\times 384$       & 129,792 FLOPs   & 448,561,152 FLOPs   & 99.97\% & 2535.00 KB & 2535.00 KB \\ \hline
  \multirow{2}{*}{\textbf{Conv-5}}&n=13, H=256   &            &                 &                     &         &            &            \\ 
                                  &k=3, s=1      & $13\times13\times 384$       & 108,160 FLOPs   & 299,040,768 FLOPs   & 99.96\% & 2112.50 KB & 2112.50 KB \\ \hline
  \textbf{FC-1}                   &m=9216,T=4096& 9216                   & 13,312 FLOPs    & 75,497,472 FLOPs    & 99.98\% & 260.00 KB  & 260.00 KB  \\ 
  \hline
  \textbf{FC-2}                  &m=4096,T=4096& 4096                   & 8,192 FLOPs     & 33,554,432 FLOPs    & 99.98\% & 160.00 KB  & 160.00 KB  \\ 
  \hline
  \textbf{FC-3}                   &m=4096,T=1000& 4096                   & 5,096 FLOPs     & 8,192,000 FLOPs     & 99.94\% & 99.53 KB   & 99.53 KB   \\ 
  \hline
  \textbf{Total Cost}             &N/A           &N/A                     & 1,074,307 FLOPs & 2,270,512,192 FLOPs & 99.95\% & 20.49 MB   & 20.49 MB   \\ \hline
  \multicolumn{8}{|c|}{\textbf{Computation for All Activation and Pooling Layers on the IoT: 650,080 FLOPs and 1,102,176 FLOPs}} \\ \hline
\end{tabular}
\vspace{-5mm}
\end{center}
\end{table*}

\begin{table*}
\caption{Example Comparison with/without Integrity Check}\label{t:comparison-integrity-check}
\vspace{-0.2cm}
\begin{center}
\begin{tabular}{ |c |c |c | c |c | c | c |c | c |}
\hline
& \multicolumn{2}{c|}{\textbf{Computation of the IoT}} & \multicolumn{2}{c|}{\textbf{Offloaded Percentage}} &\multicolumn{2}{c|}{\textbf{Communication Cost}}& \multicolumn{2}{c|}{\textbf{Storage Overhead}} \\ \cline{2-9} 
\multirow{2}{*}& {\textbf{No Integrity}}&\textbf{Integrity}&\textbf{No Integrity}&\textbf{Integrity} &\textbf{No Integrity}&\textbf{Integrity}&\textbf{No Integrity}&\textbf{Integrity} \\
&\textbf{Check}&\textbf{Check}&\textbf{Check}&\textbf{Check}&\textbf{Check}&\textbf{Check}&\textbf{Check}&\textbf{Check} \\ \hline
\textbf{Conv-1} &444,987 FLOPs&866,793 FLOPs  &99.79\%&99.59\% &8691.15 KB&8691.27 KB&8691.15 KB&8918.03 KB \\ \hline
\textbf{Conv-2} &256,608 FLOPs&2,944,608 FLOPs&99.97\%&99.67\% &5011.88 KB&5011.98 KB&5011.88 KB&5136.88 KB \\ \hline
\textbf{Conv-3} &108,160 FLOPs&2,504,320 FLOPs&99.96\%&99.16\% &2112.50 KB&2112.60 KB&2112.50 KB&2180.00 KB \\ \hline
\textbf{Conv-4} &129,792 FLOPs&3,724,032 FLOPs&99.97\%&99.17\% &2535.00 KB&2535.10 KB&2535.00 KB&2602.50 KB \\ \hline
\textbf{Conv-5} &108,160 FLOPs&3,398,272 FLOPs&99.96\%&98.86\% &2112.50 KB&2112.59 KB&2112.50 KB&2157.50 KB \\ \hline

\end{tabular}
\vspace{-5mm}
\end{center}
\end{table*}

\subsection{Numerical Analysis}
The numerical analysis of LEP-CNN is summarized in Table \ref{t:na-summ}. For expression simplicity, we use one floating point operation \textbf{\textit{FLOP}} to denote an addition or a multiplication. For a general convolutional layer, $n\times n\times D$ is the size of input, $s$ is the stride, $p$ is the size of padding, $H$ is the number of kernels, and $k\times k$ is the size of kernel matrix. For a general fully-connected layer, $m$ is the dimension of the input vector, $T$ is the number of neurons. For a pooling layer, $q\times q$ is the size of pooling regions. We use an uncompressed AlexNet \cite{AlexNet} architecture as the study case for analysis, which is a complex CNN architecture that requires 2.27 billion FLOPs for each inference request, which has comparable computing loads as the current prevalent FaceNet (1.6 billion FLOPs) \cite{FaceNet} and ResNet (3.6 billion FLOPs) \cite{ResNet}.

\subsubsection{Computational Cost}\label{sss:comp-cost}
In the \textit{Online} phase of LEP-CNN, the IoT device offloads compute-intensive convolutional layers and fully-connected layers to the edge devices. Given a general convolutional layer, the IoT device only needs to perform $D$ matrix addition with $Dn^2$ FLOPs for encryption and $H(\frac{n-k+2p}{s}+1)^2$ FLOPs for decryption respectively. Compared with executing the same convolutional layer fully on the IoT device, which takes $2DHk^2(\frac{n-k+2p}{s}+1)^2$ FLOPs, LEP-CNN significantly reduces real-time computation on the IoT device. It is worth to note that the stride $s$ in a convolutional layer is typically a small value (e.g., 1 or 2). For a general fully-connected layer, the IoT device needs to perform $m$ FLOPs for encryption and $T$ FLOPs for decryption as shown in Eq.\ref{e:ppfl-enc} and Eq.\ref{e:ppfl-dec} respectively. Differently, if the IoT device executes such a fully-connected layer at local, $2mT$ FLOPs are needed. 

Besides the offloading of convolutional layers and fully-connected layers, the IoT device also needs to process non-linear layers at local. Fortunately, these non-linear layers are extremely compute-efficient. Taking the widely adopted activation layer - ReLU as an example, it only requires $\frac{1}{2Dk^2}$ of its previous convolutional layer's cost, and $\frac{1}{2m}$ of the cost if the previous layer is a fully-connected layer.

We now discuss the computational cost of LEP-CNN using AlexNet. As shown in Table \ref{t:AlexNet}, LEP-CNN can offload over 99.9\% computational cost for convolutional layers and fully-connected layers, and only leaves lightweight encryption and decryption on the IoT device.  Compared with the offloaded convolutional layers and fully-connected layers, the local execution of all non-linear layers only requires $0.08\%$ operations for AlexNet. This result further affirms our motivation to offload convolutional layers and fully-connected layers. 

With regards to the encrypted execution on the edge device, LEP-CNN achieves the same computational cost as that directly using unencrypted data as shown in Table \ref{t:na-summ}. This is because our encryption (Eq.\ref{e:ppcl-enc} and Eq.\ref{e:ppfl-enc}) in \textit{PPCL} and \textit{PPFL} schemes make the ciphertexts $Enc(\mathcal{I}_d)$ and $Enc(\mathcal{V})$ remain the same dimension as their plaintexts $\mathcal{I}_d$ and $\mathcal{V}$. Such a decent property guarantees real-time computational performance on the edge device.

In the \textit{Offline} phase, the IoT device owner first prepares encryption keys by choosing random matrices for convolutional layers and fully-connected layers that will be offloaded. Meanwhile, the owner will take these encryption keys as inputs for their corresponding convolutional layers or fully-connected layers to obtain results as the decryption keys. In Section \ref{s:evaluation}, we show that the offline phase can be efficiently executed using a regular laptop.

\subsubsection{Communication Cost}
The communication cost of LEP-CNN comes from the transmission of encrypted inputs and outputs of convolutional layers and fully-connected layers. In our implementation, we use 160-bit random numbers (i.e., $\lambda=160$) during all encryption processes in Eq.\ref{e:ppcl-enc} and Eq.\ref{e:ppfl-enc}. Thus, each element in the ciphertext (a matrix or a vector) is 20-Byte long. To offload a convolutional layer with a $n\times n \times D$ input, the IoT device first sends its corresponding ciphertext contains  $D$ encrypted matrices with $Dn^2$ elements in total. Then, $H$ encrypted result matrices are received from the edge device with each size of $(\frac{n-k+2p}{s}+1)\times (\frac{n-k+2p}{s}+1)$. With regards to the offloading of a fully-connected layer that takes a $m$-dimensional vector as input, the IoT device needs to send a $m$-dimensional vector as encrypted input and receive a $T$-dimensional vector as encrypted output from the edge device. 

As shown in Table \ref{t:AlexNet}, the communication cost for an offloading of the AlexNet is 20.49MB, which can be efficiently handled under the edge computing environment \cite{5G-MEC}.

\begin{table*}
\caption{Experimental Evaluation Results on AlexNet}\label{t:AlexNet-IoT-Comparison}
\vspace{-0.2cm}
\begin{center}
\begin{tabular}{ |c |c |c |c |c |c |c|}
  \hline
                      &\textbf{IoT without}&\multicolumn{5}{c|}{\textbf{LEP-CNN}}\\ \cline{3-7}
                     &\textbf{Offloading}&\textbf{IoT Computation}&\textbf{Edge Computation}&\textbf{Communication}&\textbf{Total}&\multirow{2}{*}{\textbf{Speedup}} \\    
                     & \textbf{(second)}        & \textbf{(second)}             & \textbf{(second)}& \textbf{(second)} &\textbf{(second)}& \\ \hline
  \textbf{Conv-1}    & 10.01      & 0.037             & {0.0103}             & 0.849          & {0.896}           & 11.17$\times$    \\ \hline
  \textbf{Conv-2}    & 40.68      & 0.0405            & 0.0435            & 0.489          & {0.573}          & {70.99}$\times$    \\ \hline
  \textbf{Conv-3}    & 19.93      & 0.0437            & 0.013             & 0.206          & {0.263}           & {75.78}$\times$    \\ \hline
  \textbf{Conv-4}    & 29.78      & 0.0498            & 0.0184            & 0.248          & {0.316}           & {94.24}$\times$    \\ \hline
  \textbf{Conv-5}    & 19.88      & 0.0420            & {0.0127}            & 0.206          & {0.261}           & {76.17}$\times$    \\ \hline
  \textbf{FC-1}      & 2.22       & 0.0013            & 0.0043            & 0.025          & {0.031}          & {71.61}$\times$    \\ \hline
  \textbf{FC-2}      & 1.08       & 0.001             & {0.0025}            & 0.016          & {0.019}          & {56.84}$\times$    \\ \hline
  \textbf{FC-3}      & 0.27       & 0.0008            & {0.0009}            & 0.01          & {0.012}          & {22.5}$\times$    \\ \hline
  \textbf{Non-linear}& 1.137      & 1.137             & N/A               & N/A            & 1.137          & N/A              \\ \hline
  \textbf{Total Cost}& 124.99     & 1.353             & {0.106}             & 2.049          & {3.508}           & {35.63}$\times$    \\ \hline

\end{tabular}
\end{center}
\vspace{-6mm}
\end{table*}

\begin{table*}
\caption{Experimental Evaluation Results on Integrity Check}\label{t:experiment-integrity-check}
\vspace{-0.2cm}
\begin{center}
\begin{tabular}{ |c |c |c |c |c |c |}
  \hline
                      &\textbf{IoT without} &\multicolumn{4}{c|}{\textbf{LEP-CNN}}\\ \cline{3-6}
                      &\textbf{Offloading}  &\textbf{Integrity Check} &\multirow{2}{*}{\textbf{Speedup}} &\textbf{No Integrity Check}&\multirow{2}{*}{\textbf{Speedup}} \\    
                      & \textbf{(second)}   & \textbf{(second)}      &                  &\textbf{(second)}  & \\ \hline
  \textbf{Conv-1}    & 10.01      & 0.916             & 10.93$\times$                & {0.896}           & 11.17$\times$    \\ \hline
  \textbf{Conv-2}    & 40.68      & 0.655             & 62.11$\times$              & {0.573}          & {70.99}$\times$    \\ \hline
  \textbf{Conv-3}    & 19.93      & 0.402             & 49.58$\times$              & {0.263}           & {75.78}$\times$    \\ \hline
  \textbf{Conv-4}    & 29.78      & 0.524             & 56.83$\times$               & {0.316}           & {94.24}$\times$    \\ \hline
  \textbf{Conv-5}    & 19.88      & 0.470             & 42.30$\times$                 & {0.261}           & {76.17}$\times$    \\ \hline
  \textbf{FC-1}      & 2.22       & 0.031             & 71.61$\times$               & {0.031}          & {71.61}$\times$    \\ \hline
  \textbf{FC-2}      & 1.08       & 0.019             & 56.84$\times$                 & {0.019}          & {56.84}$\times$    \\ \hline
  \textbf{FC-3}      & 0.27       & 0.012             & 22.5 $\times$                 & {0.012}          & {22.5}$\times$    \\ \hline
  \textbf{Non-linear}& 1.137      & 1.137             & N/A                         & 1.137          & N/A              \\ \hline
  \textbf{Total Cost}& 124.99     & 4.166             & 30.00$\times$                & {3.508}           & {35.63}$\times$    \\ \hline

\end{tabular}
\end{center}
\vspace{-6mm}
\end{table*}

\subsubsection{Storage Overhead}\label{ss:storage}
For the offloading of a convolutional layer with a $n\times n \times D$ input, the IoT device needs to store $D$ random matrices with $n^2$ elements each as the encryption keys, and $H$ matrices with size of $(\frac{n-k+2p}{s}+1)\times (\frac{n-k+2p}{s}+1)$ as the decryption keys. To offload a fully connected layer with a $m$-dimensional vector as input, a $m$-dimensional vector and a $T$-dimensional vector need to be pre-stored as the encryption key and decryption key respectively. 



Table \ref{t:AlexNet} shows the offloading of an AlexNet request needs 20.49MB storage overhead. With the rise of IoT devices, low-power-consumption SD memory card has become an excellent fit to economically extend the storage of IoT devices \cite{SD-IoT}, which usually have more than 32GB capacity. 


\subsubsection{Additional Resource Consumption of Integrity Check}\label{ss:analysis-of-integrity-check}
Turning on the integrity check leads to additional resource consumption to local IoT device. As shown in Table \ref{t:Numerical-Analysis-Integrity-Check}, given a returned matrix of size $H(\frac{n-k+2p}{s}+1)^2$ and a sample rate of $r$, the validation process in Section \ref{ss:integrity-check} brings $\lceil rH(\frac{n-k+2p}{s}+1)^2 \rceil$ additional computation and makes the total computational cost of IoT devices rise to $2Dk^2\lceil rH(\frac{n-k+2p}{s}+1)^2 \rceil$. Since any convolutional result in the entire set of response map can be incorrect, IoT devices need to store all kernel parameters of each convolutional layer locally, which adds on $Hk^2$ storage overhead and makes the total IoT storage overhead to be $Dn^2 + H(\frac{n-k+2p}{s}+1)^2 + Hk^2$.


Table \ref{t:comparison-integrity-check} shows the resource consumption comparison between LEP-CNN with integrity check turned on and turned off. The results are calculated when error rate $\theta = 1\%$ and sample rate $r = 0.2\%, 0.3\%, 0.8\%, 0.8\%, 1.1\%$ for Conv-1 - Conv-5 respectively. Under this setting, IoT device can achieve $99\%+$ error detection rate in each convolutional layer. Since all the multiplication results of $r\theta$ are less equal to $1.1\times10^{-4}$, the additional communication costs resulted from integrity check are tiny. As a result, the communication increments are less than $4.74\times10^{-3}\%$ of the original communication costs. Compared with the heavy parameters in fully-connected layers, the parameters in convolutional layers only stand for a minor portion of the entire neural network model. Thus, even the highest additional storage overhead is only 227 KB while the lowest increment can be as low as 45 KB.

\section{Prototype Evaluation}\label{s:evaluation}
We implemented a prototype of LEP-CNN using Python 2.7. In our implementation, TensorFlow and Keras libraries are adopted to support CNNs. The resource-constrained IoT device is a Raspberry Pi (Model A) with Raspbian Debian 7, which has 700 MHz single-core processor, and 256MB memory, and 32GB SD card storage. The edge device and the IoT device owner is a Macbook Pro laptop with OS X 10.13.3, 3.1 GHz Intel Core i7 processor, 16GB memory, and 512GB SSD. The IoT device and the edge device are connected using WiFi in the same subnet. We use the well-known ImageNet \cite{imagenet_cvpr09} as the dataset for the evaluation of AlexNet. The security parameter $\lambda$ is set as 160 in our implementation. We also implemented an AlexNet-structured CryptoNets \cite{CryptoNets} as an example to compare our scheme with homomorphic encryption based privacy preserving neural networks \cite{CryptoNets,CryptoDL,IACR-PPDL}.



\subsection{Evaluation Results}
\textbf{\textit{Efficiency}}: In this section, we first evaluate the efficiency of offline phase of LEP-CNN, and then discuss the online phase for CNN execution efficiency. 

To generate the encryption and decryption keys for the execution of one AlexNet request, LEP-CNN only requires 114ms for the IoT device owner. While each set of keys will only be used for one request, the owner can efficiently compute more than 2600 sets of keys for AlexNet using 5 minutes. 

In the online phase, the IoT device in LEP-CNN efficiently executes a CNN request with privacy-preserving offloading to the edge device. Table \ref{t:AlexNet-IoT-Comparison} summarizes the evaluation results of LEP-CNN on AlexNet. By applying LEP-CNN, the required computational time on the IoT device is reduced to about $\frac{1}{92}$ for AlexNet. With such a high computational reduction, we not only overcome the challenges from limited computational resources of IoT devices, but also save energy consumption for them to achieve longer battery life. The other part of computational cost of LEP-CNN is from the privacy-preserving execution of convolutional layers and fully-connected layers on the edge device. As shown in the fourth column of Table \ref{t:AlexNet-IoT-Comparison}, the edge device served by a laptop can efficiently handle these operations using encrypted data. In practice, the selection of layers to offload in CNNs can be customized according to their complexity, since our \textit{PPCL} and \textit{PPFL} schemes are designed as independent modules for flexible combination. We also compare the privacy-preserving execution of CNN layers on the edge device with that using unencrypted data. Table \ref{t:comp-no-privacy} shows that the encryption execution on the edge device using our scheme spends almost the same time as executing these layers without privacy protection. This is also consistent with our numerical analysis in Section \ref{sss:comp-cost}, since our ciphertext has the same dimensions as its corresponding plaintext. 

As our scheme requires the interaction between the IoT device and the edge device during the execution of a CNN request, another part of cost of LEP-CNN is the communication cost. In our implementation, we use a wireless network with 10MB/s transmission speed between the IoT device and the edge device. In real-world scenario, the devices are likely to be connected via wired or cellular connection, which allows a higher transmission speed than our experimental environment. As presented in the fifth column of Table \ref{t:AlexNet-IoT-Comparison}, the total communication time for AlexNet in our network environment is only about $\frac{1}{61}$ compared with processing the entire AlexNet on the IoT device without LEP-CNN. Moreover, the upcoming 5G era for MEC environment will significantly empower the transmission speed \cite{5G-MEC} and further optimize the communication performance our scheme.

We now compare the total cost of LEP-CNN with directly executing CNN on the IoT device. As shown in Table \ref{t:AlexNet-IoT-Comparison}, LEP-CNN can speed up the execution of an AlexNet request for 35.63$\times$. Among convolutional layers and fully-connected layers in AlexNet, LEP-CNN can speed up the execution for over 90$\times$. In Table \ref{t:experiment-integrity-check}, when the integrity check is turned on, LEP-CNN can still achieve a high speedup rate of $30.00\times$ compared with AlexNet local execution. These results also validate the scalability of LEP-CNN. More to mention, with increasing complexity of convolutional layers and fully-connected layers, LEP-CNN retains or increases the high speedup rate as shown in the last column of Table \ref{t:AlexNet-IoT-Comparison} and \ref{t:experiment-integrity-check}. Taking AlexNet as example, the highest speedup rates for them are all achieved with these more complex layers. Therefore, LEP-CNN is promising to be scaled to support more complex CNN architectures according to practical requirements.

\begin{table}
\caption{Executing Each Layer of AlexNet using LEP-CNN and Non-privacy-preserving Approach on the Edge}\label{t:comp-no-privacy}
\vspace{-0.2cm}
\begin{center}
\begin{tabular}{ |c |c |c |}
\hline
  &\textbf{LEP-CNN}&\textbf{No Privacy Protection} \\ 
  &\textbf{(Second)} &\textbf{(Second)} \\ \hline
  \textbf{Conv-1}    & 0.014  & 0.012 \\ \hline
  \textbf{Conv-2}    & 0.0435 & 0.041 \\ \hline
  \textbf{Conv-3}    & 0.013  & 0.013 \\ \hline
  \textbf{Conv-4}    & 0.0184 & 0.016 \\ \hline
  \textbf{Conv-5}    & 0.012  & 0.012 \\ \hline
  \textbf{FC-1}      & 0.0043 & 0.004 \\ \hline
  \textbf{FC-2}      & 0.0065 & 0.0063 \\ \hline
  \textbf{FC-3}      & 0.0022 & 0.002 \\ \hline
  \textbf{Total Cost}& 0.123  & 0.106 \\ \hline

\end{tabular}
\vspace{-0.5cm}
\end{center}
\end{table}

\textbf{\textit{Energy Consumption}}: Compared with fully executing AlexNet inference task on the IoT device with high energy consumption, LEP-CNN significantly saves the energy consumption for computation of the IoT device while introducing slight extra energy consumption for communication. In our evaluation, the IoT device (Raspberry Pi Model A) is powered by a 5V micro-USB adapter. The voltage and current is measured using a Powerjive USB multimeter \cite{powerjive}. Table \ref{t:energy-consumption} shows the average IoT power consumption under different IoT device status. We observe that the network connection is a major power cost in IoT device. An idle IoT device with network connection can have a comparable power cost as executing AlexNet locally without network connection. In our measurement, the average active current consumption for the IoT device is $162mA$, which indicates at least 101.24J energy consumption when fully executing an inference task on the IoT device with 124.99 seconds as stated in Table \ref{t:energy-consumption}. Differently, LEP-CNN reduces the computation on the IoT device to 1.353 seconds (1.59J energy consumption) with 2.049 seconds extra communication (2.90J energy consumption). Therefore, LEP-CNN can save IoT energy consumption by $\frac{101.19-(1.59+2.90)}{101.19}=95.56\%$.


\begin{table*}
\caption{Power and Energy Consumption Evaluation}\label{t:energy-consumption}
\vspace{-0.2cm}
\begin{center}
\begin{tabular}{ |c |c |c |c |c |}
  \hline
                &\textbf{IoT without}&\multicolumn{3}{c|}{\textbf{LEP-CNN}}\\ \cline{3-5}
\multirow{2}{*}{}&\textbf{Offloading \& }&\textbf{Idle IoT with}&\textbf{IoT}&\textbf{IoT}\\    
                & \textbf{Network Connection}   & \textbf{Network Connection}   & \textbf{Computation}& \textbf{Communication}\\ \hline
  \textbf{Power (W)}      & 0.81        & 0.78              & 1.17               & 1.42                       \\ \hline
  \textbf{Energy (J)}     & 101.19      & N/A               & 1.59               & 2.90                       \\ \hline          

\end{tabular}
\end{center}
\vspace{-6mm}
\end{table*}

\textbf{\textit{Accuracy}}: To validate that there is no accuracy loss in LEP-CNN, we also implemented original AlexNet without any encryption. By using the same parameters, LEP-CNN achieves the exact same accuracy ($80.1\%$) as that obtained using original AlexNet \cite{AlexNet} without any encryption, because there is no approximation design in LEP-CNN.

\textbf{\textit{Evaluation of Sample Rate $r$}}: In order to achieve a high error detection rate, different sample rate $r$ needs to be calculated based on different settings in each convolutional layer. As shown in Figure \ref{f:sample-rate-error-detection-rate}, to make the error detection rate to surpass $99\%$, Conv-1 - Conv-5 need to set $r$ to be $0.2\%, 0.3\%, 0.8\%, 0.8\%, 1.1\%$ respectively. Figure \ref{f:sample-rate-returned-data-size} shows that as the size of the returned data rises, the sample rate $r$ required to reach $99\%+$ error detection rate drops correspondingly. From this observation combined with Figure \ref{f:sample-rate-offloaded-computation-percentage}, the scalability of the integrity check feature is ensured and the additional resource consumption of a larger, more complex CNN is always minor compared with its original costs.

\begin{figure}[ht]
\centering
\includegraphics[width=8cm]{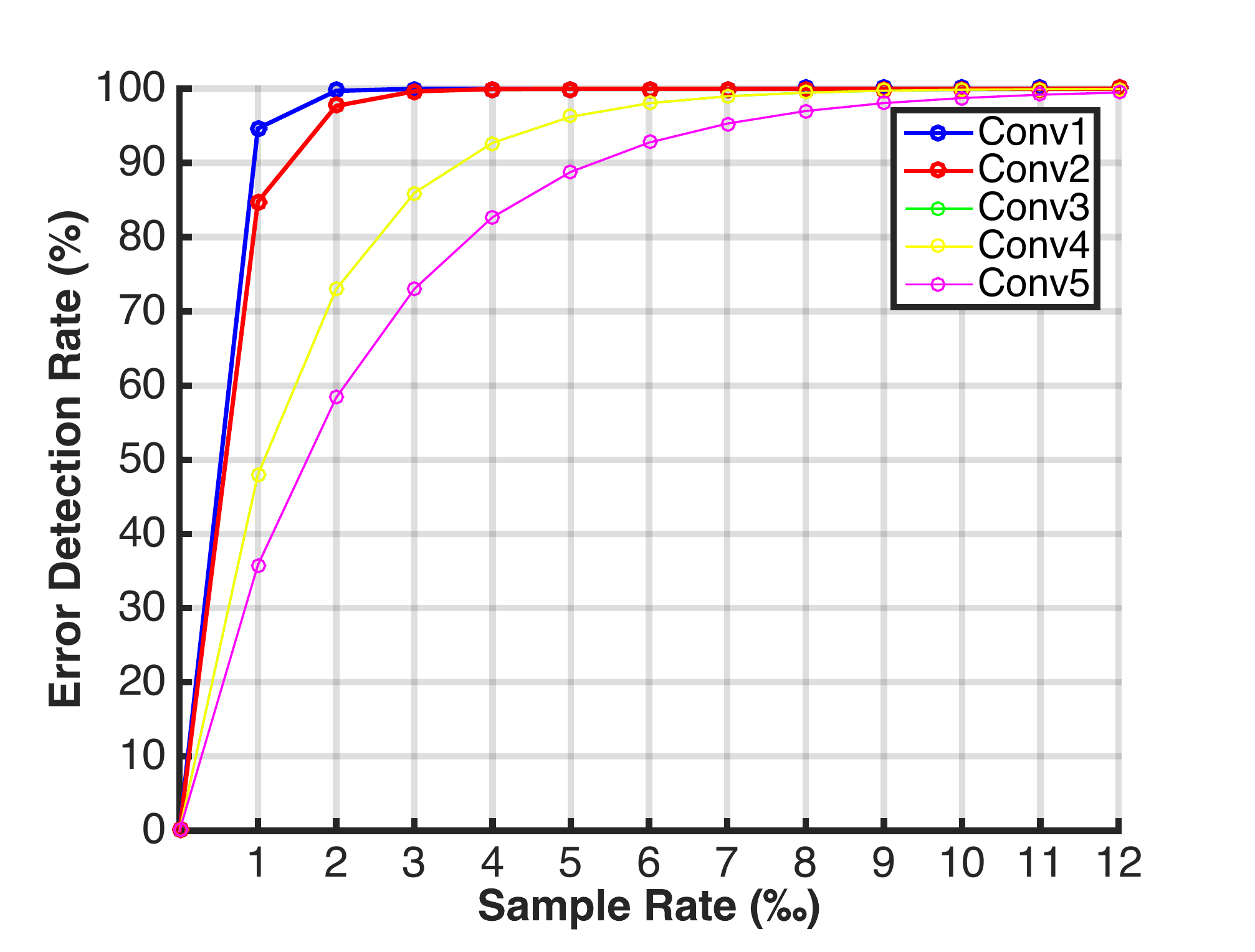}
\centering
\caption{Evaluation of Sample Rate $r$ and Error Detection Rate} \label{f:sample-rate-error-detection-rate}
\end{figure} 

\begin{figure}[ht]
\centering
\includegraphics[width=8cm]{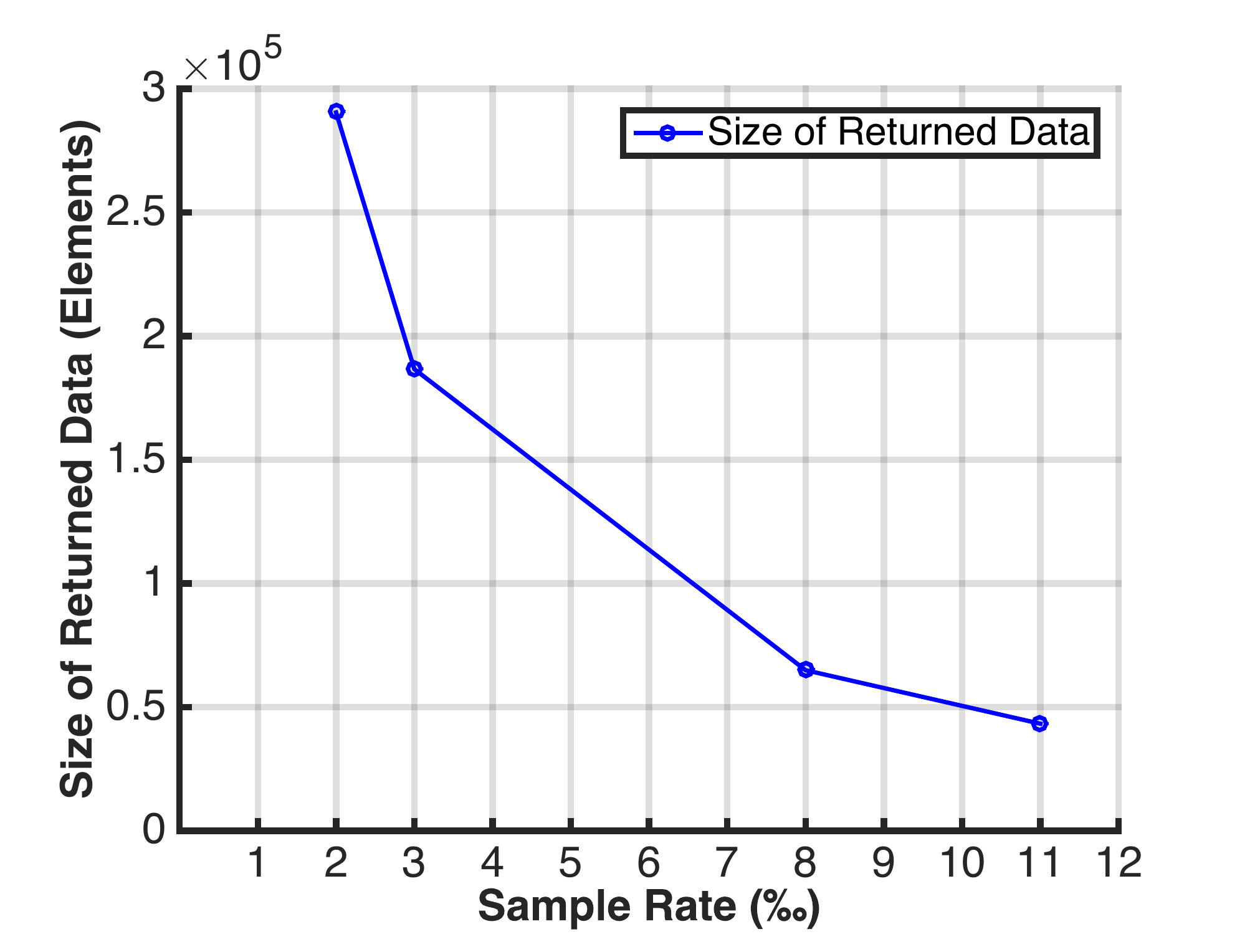}
\centering
\caption{Evaluation of Sample Rate $r$ and Returned Data Size} \label{f:sample-rate-returned-data-size}
\end{figure}

\begin{figure}[ht]
\centering
\includegraphics[width=8cm]{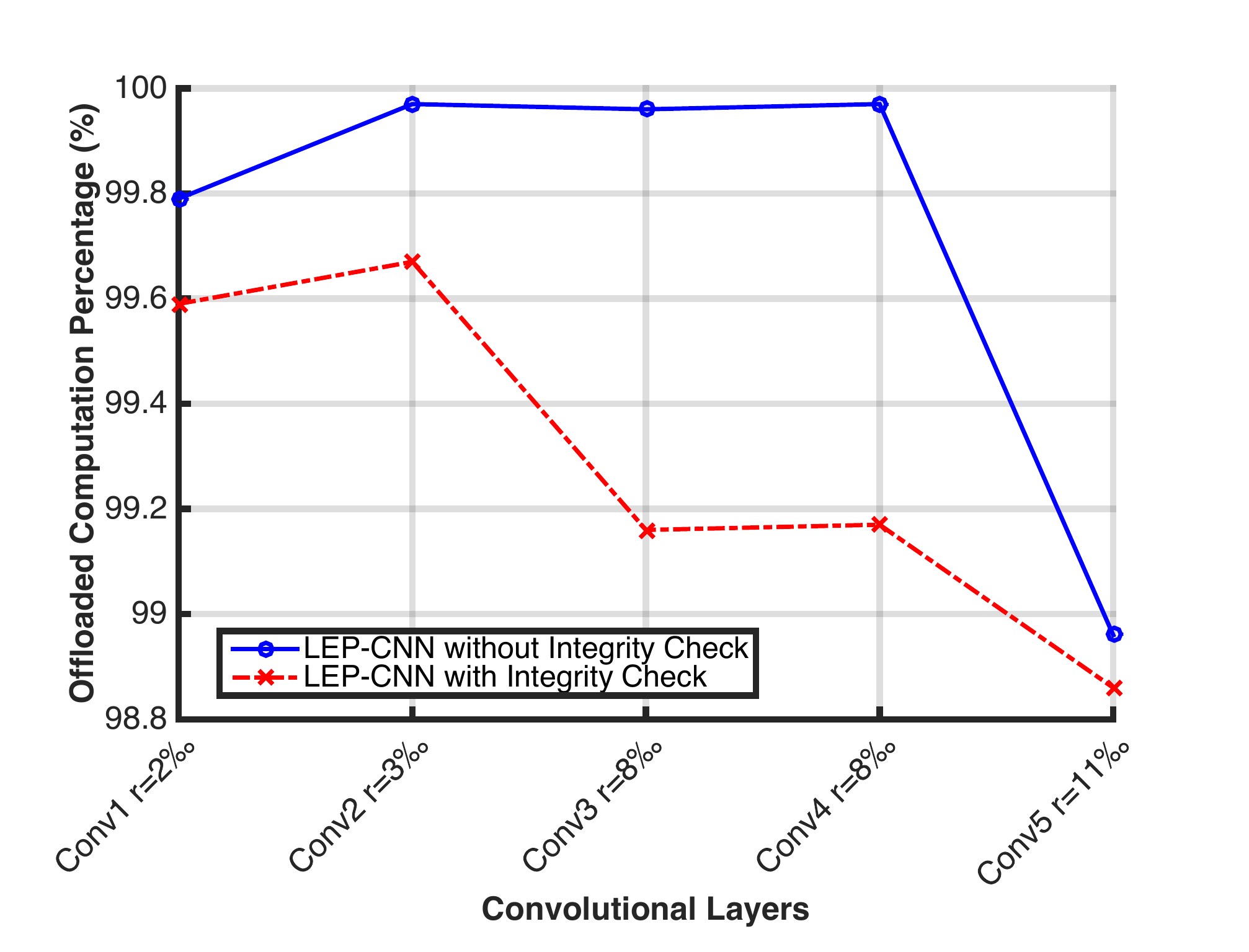}
\centering
\caption{Evaluation of Convolutional Layers and Offloaded Computation Percentage} \label{f:sample-rate-offloaded-computation-percentage}
\end{figure}

\begin{table}
\caption{Comparison Between LEP-CNN and CryptonNets in First Convolutional Layer of AlexNet}\label{t:cryptonets-comparison}
\vspace{-0.2cm}
\begin{center}
\begin{tabular}{ |c |c |c |}
\hline
                        &\textbf{LEP-CNN}     &\textbf{AlexNet Scale CryptoNets}   \\
                        &\textbf{(seconds)}   &\textbf{(seconds)}                  \\ \hline
  \textbf{Encryption}   & 0.012               & 0.360                              \\ \hline
  \textbf{Convolution}  & 0.0103              & 625.856                            \\ \hline
  \textbf{Decryption}   & 0.025               & 0.271                              \\ \hline

\end{tabular}
\vspace{-0.5cm}
\end{center}
\end{table}


\textbf{\textit{Comparison with CryptoNets}}: To compare LEP-CNN with homomorphic encryption based privacy preserving neural networks \cite{CryptoNets,CryptoDL,IACR-PPDL} under the same scale, we implemented AlexNet in CryptoNets version, denoted as A-CryptoNets, using the same network settings as in \cite{AlexNet} and same YASHE cryptosystem \cite{YASHE} as in \cite{CryptoNets}. Table \ref{t:cryptonets-comparison} shows the efficiency of the first convolutional layer in LEP-CNN and A-CryptoNets. Due to the large input size of AlexNet, the polynomials in A-CryptoNets needs to be big enough to prevent multiplication overflow. These big polynomials make the first convolutional operation cost A-CryptoNets over 10 minutes, which is even more than the total time cost of walking through each layer in LEP-CNN. Although processing requests in batches can help relieve the high computation cost of each inference request in A-CryptoNets, additional time cost is introduced to collect and form single requests into a batch. Thus, in time-sensitive scenarios, LEP-CNN has a better performance than homomorphic encryption based solutions. 


\section{Related Work}\label{s:related-work}

To enable the offloading of the CNN in a privacy-preserving manner, Gilad-Bachrach et al. \cite{CryptoNets} propose a CryptoNets using ``YASHE'' homomorphic encryption. CryptoNets allows cloud computing servers to perform the inference stage of a CNN using encrypted data only. Meanwhile, the activation function in the activation layers is replaced with the squared function to integrate homomorphic encryption. However, high computational and communication cost are introduced in CryptoNets due to the utilization of homomorphic encryption, and thus making them inefficient for time-sensitive application. In addition, the encryption cost for each request using homomorphic encryption is also expensive for resource-constrained IoT devices. Another limitation of CryptoNets is that its effectiveness can only be guaranteed for small number of activation layers as indicated in ref \cite{IACR-PPDL}. As a result, only small scale CNN architectures can be supported. To improve CryptoNets, Chabanne et al. \cite{IACR-PPDL} utilize low-degree polynomials to approximate activation layers. In addition, a normalize layer is added before the non-linear layer with batch normalization, with which polynomial approximation only needs to be accurate on a small and fixed interval. Nevertheless, ref \cite{IACR-PPDL} follows the same homomorphic encryption-based design for privacy protection as that in CryptoNets, and also suffers from the high computational and communication cost. Recently, ref \cite{CryptoDL} introduces CryptoDL that enhances CryptoNets in terms of efficiency and accuracy. In CryptoDL, low-degree polynomial-based approximation designs are proposed to support commonly used activation functions (i.e. ReLU, Sigmoid, and Tanh) in activation layers. While CryptoDL reduces about $50\%$ run time compared with CryptoNets, it still requires high local encryption cost on the IoT device. Furthermore, time-sensitive tasks require IoT devices to process data on-the-fly. Unfortunately, these existing research \cite{CryptoNets,IACR-PPDL,CryptoDL} are more suitable for the ``Data Collection and Post-Processing'' routine, since they require the batch processing of a large number of requests to improve efficiency. 




\section{Conclusion}\label{s:conclusion}
In this paper, we proposed LEP-CNN that enables resource constrained IoT devices to efficiently execute CNN requests with privacy protection. LEP-CNN uniquely designs a lightweight online/offline encryption scheme. By discovering the fact that linear operations in CNNs over input and random noise can be separated, LEP-CNN can pre-compute decryption keys to remove random noise and thus boosting the performance of real-time CNN requests. By integrating local edge devices, LEP-CNN ameliorates the network latency and service availability issue. LEP-CNN also makes the privacy-preserving operation on the edge device as efficient as that on unencrypted data. Moreover, the privacy protection in LEP-CNN does not introduce any accuracy loss to the CNN inference. LEP-CNN also provides optional integrity check functionality to help IoT devices detect erroneous results from dishonest edge devices. Thorough security analysis is provided to show that LEP-CNN is secure in the defined threat model. Extensive numerical analysis as well as prototype implementation over the well-known CNN architectures and datasets demonstrate the practical performance of LEP-CNN. Our experimental results also depict that LEP-CNN prevails in terms of accuracy and efficiency under time-sensitive scenarios compared with homomorphic encryption based offloading solutions.



\bibliographystyle{unsrt}
\bibliography{CNS-19-LEPCNN-Journal}

\begin{thebibliography}{10}

\bibitem{dl-iot-1}
M.~Verhelst and B.~Moons.
\newblock Embedded deep neural network processing: Algorithmic and processor
  techniques bring deep learning to iot and edge devices.
\newblock {\em IEEE Solid-State Circuits Magazine}, 9(4):55--65, Fall 2017.

\bibitem{dl-iot-2}
S.~Kodali, P.~Hansen, N.~Mulholland, P.~Whatmough, D.~Brooks, and G.~Y. Wei.
\newblock Applications of deep neural networks for ultra low power iot.
\newblock In {\em 2017 IEEE International Conference on Computer Design
  (ICCD)}, pages 589--592, Nov 2017.

\bibitem{dl-iot-5}
Matt Burns.
\newblock {Arm chips with Nvidia AI could change the Internet of Things}.
\newblock
  \url{https://techcrunch.com/2018/03/27/arm-chips-will-with-nvidia-ai-could-change-the-internet-of-things/},
  2018.
\newblock [Online; accessed July-2018].

\bibitem{dl-iot-survey}
Mohammadi Mehdi, Al-Fuqaha Ala, Sorour Sameh, and Guizani Mohsen.
\newblock Deep learning for iot big data and streaming analytics: A survey.
\newblock {\em arXiv:1712.04301}, 2017.

\bibitem{AlexNet}
Alex Krizhevsky, Ilya Sutskever, and Geoffrey~E. Hinton.
\newblock Imagenet classification with deep convolutional neural networks.
\newblock In {\em Proceedings of the 25th International Conference on Neural
  Information Processing Systems - Volume 1}, NIPS'12, pages 1097--1105, USA,
  2012. Curran Associates Inc.

\bibitem{FaceNet}
F.~Schroff, D.~Kalenichenko, and J.~Philbin.
\newblock Facenet: A unified embedding for face recognition and clustering.
\newblock In {\em 2015 IEEE Conference on Computer Vision and Pattern
  Recognition (CVPR)}, pages 815--823, June 2015.

\bibitem{ResNet}
K.~He, X.~Zhang, S.~Ren, and J.~Sun.
\newblock Deep residual learning for image recognition.
\newblock In {\em 2016 IEEE Conference on Computer Vision and Pattern
  Recognition (CVPR)}, pages 770--778, June 2016.

\bibitem{iot-cloud-privacy}
J.~Zhou, Z.~Cao, X.~Dong, and A.~V. Vasilakos.
\newblock Security and privacy for cloud-based iot: Challenges.
\newblock {\em IEEE Communications Magazine}, 55(1):26--33, January 2017.

\bibitem{edge-computing}
W.~Shi, J.~Cao, Q.~Zhang, Y.~Li, and L.~Xu.
\newblock Edge computing: Vision and challenges.
\newblock {\em IEEE Internet of Things Journal}, 3(5):637--646, Oct 2016.

\bibitem{CryptoNets}
Ran Gilad{-}Bachrach, Nathan Dowlin, Kim Laine, Kristin~E. Lauter, Michael
  Naehrig, and John Wernsing.
\newblock Cryptonets: Applying neural networks to encrypted data with high
  throughput and accuracy.
\newblock In {\em Proceedings of the 33nd International Conference on Machine
  Learning, {ICML} 2016, New York City, NY, USA, June 19-24}, pages 201--210,
  2016.

\bibitem{IACR-PPDL}
Herv{\'e} Chabanne, Amaury de~Wargny, Jonathan Milgram, Constance Morel, and
  Emmanuel Prouff.
\newblock Privacy-preserving classification on deep neural network.
\newblock {\em IACR Cryptology ePrint Archive}, 2017:35, 2017.

\bibitem{CryptoDL}
Ehsan Hesamifard, Hassan Takabi, and Mehdi Ghasemi.
\newblock Cryptodl: Deep neural networks over encrypted data.
\newblock {\em CoRR}, abs/1711.05189, 2017.

\bibitem{paillier-benchmark}
Cornejo Mario and Poumeyrol Mathieu.
\newblock {Benchmarking Paillier Encryption}.
\newblock
  \url{https://medium.com/snips-ai/benchmarking-paillier-encryption-15631a0b5ad8},
  2018.
\newblock [Online; accessed July-2018].

\bibitem{CCS16}
Martin Abadi, Andy Chu, Ian Goodfellow, H.~Brendan McMahan, Ilya Mironov, Kunal
  Talwar, and Li~Zhang.
\newblock Deep learning with differential privacy.
\newblock In {\em Proceedings of the 2016 ACM SIGSAC Conference on Computer and
  Communications Security}, CCS '16, pages 308--318, New York, NY, USA, 2016.
  ACM.

\bibitem{ICDM17-PHAN}
Phan NhatHai, Wu~Xintao, Hu~Han, and Dou Dejing.
\newblock Adaptive laplace mechanism: Differential privacy preservation in deep
  learning.
\newblock In {\em Proceedings of the 2017 IEEE International Conference on Data
  Mining}, ICDM '17. IEEE, 2017.

\bibitem{edge-iot-2}
Q.~Zhang, Q.~Zhang, W.~Shi, and H.~Zhong.
\newblock Firework: Data processing and sharing for hybrid cloud-edge
  analytics.
\newblock {\em IEEE Transactions on Parallel and Distributed Systems},
  PP(99):1--1, 2018.

\bibitem{edge-iot-3}
G.~Ananthanarayanan, P.~Bahl, P.~Bodík, K.~Chintalapudi, M.~Philipose,
  L.~Ravindranath, and S.~Sinha.
\newblock Real-time video analytics: The killer app for edge computing.
\newblock {\em Computer}, 50(10):58--67, 2017.

\bibitem{5G-MEC}
B.~P. Rimal, D.~P. Van, and M.~Maier.
\newblock Mobile edge computing empowered fiber-wireless access networks in the
  5g era.
\newblock {\em IEEE Communications Magazine}, 55(2):192--200, February 2017.

\bibitem{iot-survey}
W.~Yu, F.~Liang, X.~He, W.~G. Hatcher, C.~Lu, J.~Lin, and X.~Yang.
\newblock A survey on the edge computing for the internet of things.
\newblock {\em IEEE Access}, 6:6900--6919, 2018.

\bibitem{kandukuri2009cloud}
Balachandra~Reddy Kandukuri, Atanu Rakshit, et~al.
\newblock Cloud security issues.
\newblock In {\em Services Computing, 2009. SCC'09. IEEE International
  Conference on Services Computing}, pages 517--520. IEEE, 2009.

\bibitem{imagenet_cvpr09}
J.~Deng, W.~Dong, R.~Socher, L.-J. Li, K.~Li, and L.~Fei-Fei.
\newblock {ImageNet: A Large-Scale Hierarchical Image Database}.
\newblock In {\em CVPR09}, 2009.

\bibitem{cnn-wiki}
{Wikipedia}.
\newblock {Convolutional neural network }.
\newblock \url{https://en.wikipedia.org/wiki/Convolutional_neural_network}.
\newblock [Online; accessed July-2018].

\bibitem{Cong2014}
Jason Cong and Bingjun Xiao.
\newblock {\em Minimizing Computation in Convolutional Neural Networks}, pages
  281--290.
\newblock Springer International Publishing, Cham, 2014.

\bibitem{output-inference1}
Aravindh Mahendran and Andrea Vedaldi.
\newblock Understanding deep image representations by inverting them.
\newblock {\em CoRR}, abs/1412.0035, 2014.

\bibitem{Crpto-book-cp3.3}
Jonathan Katz and Yehuda Lindell.
\newblock {\em Chapter 3.3, Introduction to Modern Cryptography}.
\newblock Chapman \& Hall/CRC, 2007.

\bibitem{uav-cloud}
J.~Lee, J.~Wang, D.~Crandall, S.~Šabanović, and G.~Fox.
\newblock Real-time, cloud-based object detection for unmanned aerial vehicles.
\newblock In {\em 2017 First IEEE International Conference on Robotic Computing
  (IRC)}, pages 36--43, April 2017.

\bibitem{uav-battery}
{Airdata UAV}.
\newblock {Drone Flight Stats}.
\newblock \url{https://airdata.com/blog/2017/drone-flight-stats-part-1}, 2018.
\newblock [Online; accessed July-2018].

\bibitem{Crpto-book-cp11}
Jonathan Katz and Yehuda Lindell.
\newblock {\em Chapter 11, Introduction to Modern Cryptography}.
\newblock Chapman \& Hall/CRC, 2007.

\bibitem{SD-IoT}
{Paul, Norbury}.
\newblock {Now Trending: SD Memory Cards}.
\newblock
  \url{https://www.sdcard.org/press/thoughtleadership/180118_Now_Trending_SD_Memory_Cards.html},
  2018.
\newblock [Online; accessed July-2018].

\bibitem{powerjive}
{Raspberry Pi Dramble}.
\newblock {Power Consumption Benchmarks}.
\newblock \url{http://www.pidramble.com/wiki/benchmarks/power-consumption},
  2018.
\newblock [Online; accessed October-2018].

\bibitem{YASHE}
Joppe~W. Bos, Kristin Lauter, Jake Loftus, and Michael Naehrig.
\newblock Improved security for a ring-based fully homomorphic encryption
  scheme.
\newblock In Martijn Stam, editor, {\em Cryptography and Coding}, pages 45--64,
  Berlin, Heidelberg, 2013. Springer Berlin Heidelberg.

\end{thebibliography}

\end{document}